\documentclass[groupedaddress,10pt]{article}
\usepackage[utf8]{inputenc}
\usepackage{amsmath,amsthm,amsfonts,latexsym,amssymb,bm,enumerate}
\usepackage{graphicx}
\usepackage{ae}
\usepackage{cite}
\usepackage{float}
\usepackage{lmodern}
\usepackage[T1]{fontenc}
\usepackage{epsf}
\usepackage{psfrag}
\usepackage{mathrsfs}
\usepackage{authblk}
\usepackage{hyperref}
\DeclareGraphicsExtensions{.eps,.art,.ART,.ps}
\newtheorem{theorem}{Theorem}
\newtheorem{Thm}{Theorem}[section]
\newtheorem{proposition}[Thm]{Proposition}

\newtheorem{lemma}[Thm]{Lemma}

\newtheorem{Remark}[Thm]{Remark}
\newtheorem{Remarks}[Thm]{Remarks}

\newcounter{mnotecount}[section]

\renewcommand{\themnotecount}{\thesection.\arabic{mnotecount}}

\newcommand{\mnote}[1]
{\protect{\stepcounter{mnotecount}}$^{\mbox{~\footnotesize
$
\bullet$\themnotecount}}$ \marginpar{
\raggedright\tiny\em
$\!\!\!\!\!\!\,\bullet$\themnotecount: #1} }

\begin{document}
\title{The formation of trapped surfaces in the gravitational collapse of spherically symmetric  scalar fields with a positive cosmological constant.}
\author{Jo\~ao L. Costa}
\affil{Mathematics Department, Lisbon University Institute -- ISCTE}
\affil{CAMGSD, Instituto Superior T\'ecnico, ULisboa}

\maketitle
%
%
\begin{abstract}

Given spherically symmetric characteristic initial data for the Einstein-scalar field system with a positive cosmological constant, we provide a criterion, in terms of the dimensionless size and dimensionless renormalized mass content of an annular region of the data,  for the formation of a future trapped surface. This corresponds to an extension of Christodoulou's classical criterion by the inclusion of the cosmological term.
\end{abstract}
%
%
%
%
%
\section{Introduction and main result}\label{sectionIntro}

The existence of a (future) trapped surface - a compact codimension 2 spacelike submanifold with negative inward and outward (future) null expansions - has profound implication for the global structure of spacetime: most notably, under quite general causal assumptions and fairly weak energy conditions, it guarantees the existence of a non-empty black hole region~\cite{Chrusciel09,Chrusciel18} and, by Penrose's Singularity/Incompleteness Theorem, it also implies future causal geodesic incompleteness. 

 In~\cite{Christodoulou91}, Christodoulou established a criterion for the dynamical formation of trapped spheres in the context of a characteristic initial value problem
  for the spherically symmetric Einstein-scalar field system; this was an essential step in his seminal proof of both the Weak and Strong Cosmic censorship conjectures, under appropriate regularity conditions for such system, that was completed in~\cite{Christodoulou99}. Later, in celebrated work~\cite{Christodoulou09}, Christodoulou extended his criterion to the context of the  Einstein vacuum equations without symmetry assumptions; in the mean time, important developments have been achieved in that setting~\cite{Klainerman12,Klainerman14, An12, An17}.  Returning to the context of spherically symmetric spacetimes we would like to mention that criteria for the formation of trapped surfaces were also established for solutions of the Einstein-Vlasov equations, in~\cite{Andreasson17},  and the Einstein-Euler system, in~\cite{Burtscher17}.
 Concerning the impact that the formation of a trapped surface has in the global structure of asymptotically flat spherically symmetric spacetimes we refer the reader to the qualitative results in~\cite{Dafermos05,Kommemi} and the quantitative results obtained in~\cite{An20,An202,Luk19}. 

In view of the recent developments concerning the global analysis of black hole spacetimes with a positive cosmological constant - from the remarkable proofs of the non-linear stability of the local/exterior regions of Kerr-Newman de Sitter~\cite{hintzVasy,Hintz}, to the new results concerning the structure of black hole interiors~\cite{CGNS1,CGNS2,CGNS3,CGNS4}, as well as the improved  understanding of the geometry of cosmological regions~\cite{Costa-Natario-Oliveira,schlue2,DafRen} - it seems relevant to study the formation of trapped surfaces in the presence of $\Lambda>0$.
In this paper we start this study by revisiting Christodoulou's original criterion and extending it to the context of solutions  to the Einstein-scalar field system with a positive cosmological constant.  Another motivation for this paper comes form the study of gravitational collapse in the ``hard'' phase of Christodoulou's two-phase model~\cite{Christodoulou95}, where the Einstein-Euler system reduces to the Einstein-scalar field system with cosmological constant $\Lambda=1$. 

Here we will follow the general strategy developed in~\cite{Christodoulou09} closely, but in order to do so one needs to overcome  
some new challenges created by the introduction of a positive cosmological constant. Most notably the new difficulties emanate from: the well known inadequacy of the Hawking mass~\eqref{1-mu0} in the context of solutions with a positive cosmological constant and, most importantly, the loss of the basic monotonicity properties of the gradient of the radial function~\footnote{This loss of monotonicity can be traced back to the fact that $\partial_r(1-\mu)$, as defined in~\eqref{dr1-mu}, has no definite sign if $\Lambda>0$; we stress that this is not the case for $\Lambda\leq 0$. That such a sign, or lack thereof, can have a remarkable effect on the global structure of spacetime is well known, see for instance the discussion in~\cite[Section 10]{Dafermos04}.}. The first issue is resolved by the standard replacing of the Hawking mass with its renormalized version~\eqref{piDef}, in terms of which our criterion is formulated~\eqref{eta0}. 
The second difficulty is dealt with by obtaining ``weak monotonicity'' estimates, see for example~\eqref{lambdaMonot}; unfortunately their ``weakness'' propagates to the remaining estimates and makes the subsequent analysis more involved; for instance, as a consequence, we are no longer able to rely on some elementary integration formulas that are available for the $\Lambda=0$ case.         
Other specificities of the positive cosmological constant setting are discussed in Remarks~\ref{rmks} and Section~\ref{sectionSystem}.
We end this discussion 
by observing that  instead of relying on ``a geometric Bondi coordinate together with a null frame''~\footnote{As nicely summarized in~\cite{An20}.}, which provide the basic framework in~\cite{Christodoulou09}, we rely solely on a global double null coordinate setup.   

The main result of this paper is the following:    

\begin{theorem}
\label{mainThm} Let $({\cal M}={\cal Q}\times_r \mathbb{S}^2,g,\phi)$ be a smooth spherically symmetric  solution of the Einstein-scalar field system~\eqref{einstein}-\eqref{energymomentumtensor} with a cosmological constant $\Lambda>0$. Assume there exist (double null) coordinates 
$(u,v):{\cal Q}\rightarrow\mathbb{R}^2$, with $\partial_u$ and $\partial_v$ future oriented and 
\begin{equation}
 \label{metric}
 g= - \Omega^2 (u,v)  du  dv + r^2 (u,v)  \mathring{g}\;,
\end{equation}
where $\mathring{g}$ is the metric of the round 2-sphere. 
Assume also that ${\cal C}_0^+=\{u=0\;,\; v\geq 0\}$ is a future complete null line emanating from the timelike curve $\Gamma:=\{r=0\}\subset{\cal Q}$ and that ${\cal Q}^+:=J^+({\cal C}^+_0)\cap J^+(\Gamma)$ coincides with the future maximal globally hyperbolic development of the data induced on ${\cal C}_0^+$. 
 
Let $0<v_1<v_2$ be such that: 
\begin{enumerate}[(i)]
 \item  $(0,v_1)\in J^-(\overline{\Gamma})$, where $\overline{\Gamma}$ is the closure of $\Gamma$ in $\mathbb{R}^2$, 
 \item $r(0,v_2)<1/\sqrt{\Lambda}$,
 \item $\partial_ur(0,v)<0$ and  $\partial_vr(0,v)>0$, for all $0\leq v \leq v_2$,
\end{enumerate}
and consider the {\em dimensionless size}
\begin{equation}
\label{delta0}
  \delta_0:=\frac{r(0,v_2)-r(0,v_1)}{r(0,v_1)}\;,
\end{equation}
and the {\em dimensionless renormalized mass content}
\begin{equation}
\label{eta0}
 \eta _0:=\frac{2(\varpi(0,v_2)-\varpi(0,v_1))}{r(0,v_2)}\;,
\end{equation}
where $\varpi:{\cal Q}\rightarrow\mathbb{R}$ is the {\bf renormalized} Hawking mass (see~\eqref{piDef}). 
Then, there are constants $c_1,\delta_1>0$ such that, if $\delta_0<\delta_1$ and
\begin{equation}
\label{mainAss}
 \eta_0> c_1 \delta_0\log(1/\delta_0)\;,
\end{equation}
 there exists $u^*>0$ for which {\bf $(u^*,v_2)\in\cal Q^+$ is a marginally trapped sphere}, i.e., $\partial_vr(u^*,v_2)=0$, and all {\bf $(u,v_2)\in{\cal Q}^+$, with $u>u^*$, are trapped spheres}, i.e, $\partial_vr(u,v_2)<0$.  Moreover, $r(u^*,v_1)>0$, i.e.,  the marginally trapped sphere forms, along $v=v_2$, before (as measured by $u$) the sphere $(u,v_1)$ reachs $\overline{\Gamma}$. 
\end{theorem}

\begin{Remarks}
\label{rmks}
\,
\begin{enumerate}
 \item The presented criterion reduces formally to Christodoulou's criterion~\cite{Christodoulou91} by setting $\Lambda$ to zero. 
 \item We observe that in terms of the original Hawking mass~\eqref{1-mu0} we have 
\begin{eqnarray}
\nonumber
 \eta _0&=&\frac{2(m(0,v_2)-m(0,v_1))}{r(0,v_2)}+\frac{\Lambda}{3}\frac{r^3(0,v_1)-r^3(0,v_2)}{r(0,v_2)}\\
 \label{eta0Comp}
 &\leq & \frac{2(m(0,v_2)-m(0,v_1))}{r(0,v_2)}
\end{eqnarray}
and recall that the last quantity is the one used for $\Lambda=0$ in~\cite{Christodoulou91}.  
\item Note that, in particular, our theorem shows that a trapped surface can form from the evolution of data, posed on ${\cal C}^+_0$, which is ``arbitrarily far from being trapped''. By this we mean that, for any $\epsilon>0$, we can choose initial data satisfying the assumptions of the theorem and such that 
$$0\leq \sup_{v\leq v_2}\frac{m}{r}(0,v)<\epsilon\;.$$
\item
 Condition $(iii)$ above has no parallel in the  $\Lambda=0$ case and is related to the fact that, under general conditions~\cite[Section 3]{CostaProblem}, an {\em apparent cosmological horizon} - a causal hypersurface which is the union of marginally anti-trapped spheres, i.e., where  $\partial_ur=0$ and $\partial_vr>0$ - must intersect the initial cone ${\cal C}^+_0$ in the region $\sqrt{1/\Lambda}<r\leq \sqrt{3/\Lambda}$. 
\item Note that,  by invoking the results in~\cite{CostaMena}, condition $(i)$ above can be replaced by an appropriate smallest condition at the level of the scalar field.
\item Lower bounds for the mass and radius of the marginally trapped sphere can be obtained from~\eqref{massTrapped},~\eqref{r*Est} and~\eqref{alphaDef}.
\end{enumerate}
\end{Remarks}


\section{The Einstein-scalar field system with a positive cosmological constant, in spherical symmetry}\label{sectionSystem}

We will consider the Einstein-scalar field system in the presence of a positive cosmological constant $\Lambda$:
\begin{align}
& R_{\mu\nu} - \frac{1}{2} R\, g_{\mu\nu} + \Lambda g_{\mu\nu} = 2 T_{\mu\nu} \label{einstein} \, , \\
& T_{\mu\nu} = \partial_\mu \phi\, \partial_\nu \phi - \frac{1}{2} \partial_\alpha \phi\, \partial^\alpha \phi \, g_{\mu\nu} \;,
\label{energymomentumtensor} 
\end{align}
where a Lorentzian metric $g_{\mu\nu}$  is coupled to a  scalar field $\phi$ via the Einstein field equations~\eqref{einstein}, with energy-momentum tensor~\eqref{energymomentumtensor}. 

 We will work in spherical symmetry by assuming that ${\cal M}={\cal Q}\times \mathbb{S}^2$ and by requiring the existence of double-null coordinates $(u,v)$, on $\cal Q$, such that the spacetime metric takes the form~\eqref{metric}. Then 
the system~\eqref{einstein}-\eqref{energymomentumtensor} becomes
\begin{align}
& \partial_u\partial_vr = -\frac{\Omega^2}{4r} - \frac{\partial_ur\,\partial_vr}{r} +  \frac{\Omega^2 \Lambda r}{4} \label{wave_r} \, , \\
& \partial_u\partial_v\phi = -\,\frac{\partial_ur\,\partial_v\phi+\partial_vr\,\partial_u\phi}{r} \label{wave_phi} \, , \\
& \partial_v\partial_u\ln\Omega = -\partial_u\phi\,\partial_v\phi+\frac{\Omega^2}{4r^2}+\frac{\partial_ur\,\partial_vr}{r^2} \label{wave_Omega} \, , \\
& \partial_u \left(\Omega^{-2}\partial_u r\right) = -r\Omega^{-2}\left(\partial_u\phi\right)^2 \label{raychaudhuri_u} \, , \\
& \partial_v \left(\Omega^{-2}\partial_v r\right) = -r\Omega^{-2}\left(\partial_v\phi\right)^2\label{raychaudhuri_v} \, .
\end{align}
Since we will only probe the future of the null cone $u=0$, truncated at $v= v_2$, from now on we will redefine ${\cal Q}^+$, as defined in Theorem~\ref{mainThm}, to be
the set 
$${\cal Q}^+:=\{(u,v)\in{\cal Q}\,:\,u\geq 0\;\text{ and }0\leq v\leq v_2\}\;.$$

We define the Hawking mass $m=m(u,v)$ by
\begin{equation}
\label{1-mu0}
1-\mu:=1-\frac{2m}{r} =\partial^{\alpha}r\partial_{\alpha} r\;.
\end{equation}
It turns out to be convenient, both by its physical relevance and by its good monotonicity properties, to also introduce the renormalized Hawking mass $\varpi=\varpi(u,v)$ defined by
\begin{equation}
\label{piDef}
\varpi:=m-\frac{\Lambda}{6}r^3\;.
\end{equation}
Then we have 
\begin{equation}
\label{1-mu}
1-\mu= 1-\frac{2\varpi}{r}-\frac{\Lambda}{3}r^2\;,
\end{equation}
and by interpreting $1-\mu$ as a function of $(r,\varpi)$ we write 
\begin{equation}
\label{dr1-mu}
\partial_r(1-\mu)= \frac{2\varpi}{r^2}-\frac{2\Lambda}{3}r\;.
\end{equation}

In this paper we will mainly use a first order formulation of the Einstein-scalar field system, obtained by introducing the quantities:
\begin{align}
& \nu \, := \, \partial_u r \, ,\\
& \lambda \, := \, \partial_v r \, ,\\
& \theta \, := \, r \partial_v \phi \, ,\\
& \zeta \, := \, r \partial_u \phi \, .\\
\end{align}
Note, for instance, that in terms of these new quantities~\eqref{1-mu} gives
\begin{equation}
\label{1-mu2}
1-\mu=-4\Omega^{-2} \lambda\nu\,.
\end{equation}

Consequently, in the {\em non-trapped} region 
\begin{equation}
\label{RDef}
 {\cal R}:=\{(u,v)\in{\cal Q}^+\; : \; \lambda(u,u)>0\text{ and }\nu(u,v)<0\}\;,
\end{equation}
we are allowed to define both
\begin{equation}
\label{k}
\kappa:=-\frac{1}{4}\Omega^{2}\nu^{-1}
\end{equation}
and
\begin{equation}
\label{bark}
\underline\kappa:=-\frac{1}{4}\Omega^{2}\lambda^{-1}\;.
\end{equation}

It is then well known (see for instance~\cite{CGNS1}) that the Einstein-scalar field system with a cosmological constant satisifies the following overdetermined system of PDEs and algebraic equations: 
\begin{align}
& \partial_u \lambda \, = \, \partial_v \nu \, = \,  \lambda\nu \frac{\partial_r \left(1-\mu\right)}{1-\mu} \label{dr} \, , \\ & \partial_u \varpi \, = \, \frac{\zeta^2}{2\underline\kappa} \label{dupi} \, , \\ 
& \partial_v \varpi \, = \, \frac{\theta^2}{2\kappa}  \label{dvpi} \, , \\ 
& \partial_u \theta \, = \, - \frac{\zeta \lambda}{r} \, , \label{dt} \\
& \partial_v \zeta \, = \, - \frac{\theta \nu}{r} \, , \label{dz} \\
& \partial_v \underline\kappa \, = \, \frac{\underline\kappa \theta^2}{r\lambda} \label{dkbar} \, , \\
& \partial_u \kappa \, = \, \frac{\kappa \zeta^2}{r\nu} \label{dk} \, , \\
& \nu \, = \, \underline\kappa \left(1-\mu\right) \, , \label{kbarAlg} \\
& \lambda \, = \, \kappa \left(1-\mu \right) \, . \label{kAlg}
\end{align}

It will also be useful to note that in terms of the (original) Hawking mass~\eqref{dvpi} reads
\begin{equation}
\label{dvm}
\partial_v m \, = \, \frac{\Lambda}{2} r^2\lambda+\frac{\theta^2}{2\kappa}  \;.
\end{equation}

One of the assumptions of Theorem~\ref{mainThm} is the non-existence of anti-trapped spheres along the initial null cone $\{u=0\}$, up to $v\leq v_2$, i.e, that $\nu(0,v)<0$, for all $v\leq v_2$ (recall point $4$ of Remarks~\ref{rmks}). It is then a well known consequence of the Raychaudhuri equation~\eqref{raychaudhuri_u} that this sign is preserved by evolution along the ingoing direction, that is 
\begin{equation}
\label{signNu}
 \nu(u,v)<0\;\;,\; \text{ for all } (u,v)\in{\cal Q}^+\;.
\end{equation}
As an immediate consequence we get the global upper bound
\begin{equation}
\label{rBound}
 r(u,v)\leq r(0,v_2)\;\;,\; \text{ for all } (u,v)\in{\cal Q}^+\;.
\end{equation}
Most importantly, the  following extension criteria then follows  from~\cite[Corollary 5.5]{CGNS1}:
\begin{proposition}
 \label{propExt}
Under the previous conditions (including the sign condition~\eqref{signNu}), 
let $p\in \overline{{\cal Q}^+}$ be such that there exists $q\in J^-(p)$ for which
$${\cal D}:=\big(J^-(p)\cap J^+(q)\big)\setminus\{p\}\subset {\cal Q}^+\;.$$
If there exists $c>0$ such that
 $$ inf_{\cal D} \,r>c\;,  $$
then $p\in {\cal Q}^+\setminus{\Gamma}$\;.
\end{proposition}
\begin{Remark}
Observe that if we consider $p=i^+$ in Schwarzschild de Sitter then we clearly have $ inf_{\cal D} \,r>c>0$ but this does not, and cannot, imply that such solution can be extended to $i^+$. The reason why the extension criteria in Proposition~\ref{propExt} does not apply here can be traced back to the fact that  $\nu$ is not strictly negative in any domain $\cal D$, of the form prescribed by Proposition~\ref{propExt}, since such a domain must alway contain a portion of the cosmological horizon.   
\end{Remark}

It is also useful to consider the {\em marginally-trapped} region
\begin{equation}
 \label{apparentDef}
 {\cal A}:=\{(u,v)\in {\cal Q}^+\;: \; \lambda(u,v)=0\}\;,
\end{equation}
and the {\em trapped} region
\begin{equation}
 \label{apparentDef}
 {\cal T}:=\{(u,v)\in {\cal Q}^+\;:\; \lambda(u,v)<0\}\;.
\end{equation}

We then have the following: 
\begin{proposition}
 \label{propGlobal}
If non-empty, $\cal A$ is a $C^1$ curve such that: 
\begin{equation}
 \label{J+A}
 J^+(\cal A) \setminus {\cal A} \subset \cal T\;,
\end{equation}
\begin{equation}
 \label{J-A}
 J^-(\cal A) \setminus {\cal A} \subset \cal R\;,
\end{equation}
and 
\begin{equation}
 \label{Rpast}
 J^-(\cal R) \subset \cal R\;.
\end{equation}
Moreover
\begin{equation}
 \label{Aaxis}
\Gamma\subset \cal R\;,
\end{equation}
and, consequently,
\begin{equation}
 \label{J-Gamma}
 {J}^-(\Gamma)\subset \cal R\;.
\end{equation}
\end{proposition}
\begin{proof}
%

Assume $(u^*,v^*)\in{\cal A}$.  
Then~\eqref{1-mu2} implies that $1-\mu(u^*,v^*)=0$ which in turn gives rise to the identities
\begin{equation}
 \label{massTrapped}
 m(u^*,v^*)=\frac{r(u^*,v^*)}{2}\;,
\end{equation}
and 
\begin{equation}
 \label{remassTrapped}
\varpi(u^*,v^*)=\frac{r(u^*,v^*)}{2}-\frac{\Lambda}{6}r^3(u^*,v^*)\;.
\end{equation}
The last identity together with~\eqref{rBound} implies that 
$$\partial_r(1-\mu)(u^*,v^*)=\frac{1}{r(u^*,v^*)}-\Lambda r(u^*,v^*) >0\;.$$
We will also need the fact that   
\begin{equation}
\label{kgeq0}
 \kappa > 0\;,
\end{equation}
which follows from~\eqref{k} and~\eqref{signNu}.
We then use~\eqref{wave_r},~\eqref{1-mu2} and~\eqref{k} to write
$$\partial_u \lambda =\nu \kappa \partial_r(1-\mu)$$ 
and conclude that 
$\partial_u\lambda (u^*,v^*)<0$; in particular, this shows that $\cal A$ is a continuously differentiable curve. The same reasoning shows that once $\lambda<0$, we get $1-\mu<0$, which implies that $\partial_r(1-\mu)>0$ and consequently $\partial_u\lambda<0$. We can then conclude that
\begin{equation}
\label{lambdaNeg}
 \lambda(u,v^*) <0\;\;,\;\text{ for all } u>u^*\;.
\end{equation}
Just as in the discussion leading to~\eqref{signNu} the fact that 
\begin{equation}
\label{lambdaNed}
 \lambda(u^*,v^*) \leq 0 \Rightarrow  \lambda(u^*,v) \leq 0 \;\;,\;\text{ for all } v>v^*\;,
\end{equation}
is an immediate consequence of the Raychaudhuri equation~\eqref{raychaudhuri_v}. 
Consequently both~\eqref{J+A} and~\eqref{J-A} follow. 

Now assume there existes $(u^*,v^*)\in {\cal A}\cap \Gamma$, i.e., $r(u^*,v^*)=\lambda(u^*,v^*)=0$: then, in view of~\eqref{lambdaNed} and the fact that $r\geq 0$, we get $r(u^*,v)=0$, for $v\geq v^*$, which is in contradiction with the causal character of the center of symmetry $\Gamma$. The remanning conclusions then follow immediately. 

\end{proof}
%

%
\section{Proof of Theorem~\ref{mainThm}}\label{sectionProof}

\begin{proof}
The proof follows by obtaining several estimates valid in the region $\cal R$, as defined in~\eqref{RDef}, that will allow us to conclude that, under the conditions of Theorem~\ref{mainThm},  the ingoing line $v=v_2$ has to exit $\cal R$ before leaving ${\cal Q}^+$.  So from now one, and unless otherwise stated, assume that we are in $\cal R$.

We start by noting that~\eqref{1-mu0} implies that  
\begin{equation}
 m|_{\Gamma}=0\;,
\end{equation}
 and then, after recalling~\eqref{kgeq0},~\eqref{dvm} shows that $\partial_vm\geq 0$  and since ${\cal Q}^+ \subset J^+(\Gamma)$ we obtain 
\begin{equation}
 m\geq 0\;.
\end{equation}
Consequently  
\begin{equation}
\label{less1}
0<1-\mu\leq 1  \;,
\end{equation}
with the first inequality, valid in $\cal R$, as a consequence of~\eqref{1-mu2}. 

The previous together with~\eqref{dkbar}, the monotonicity of the radial function in $v$ and~\eqref{dvpi}  gives (recall~\eqref{kAlg})
\begin{eqnarray}
 \nonumber
 \log\left(\frac{\underline\kappa(u,v_2)}{\underline\kappa(u,v_1)}\right)
 &=&
 \nonumber
 \int_{v_1}^{v_2} \frac{ \theta^2}{r\lambda}(u,\tilde v)d\tilde v \\ 
 &\geq& 
 \nonumber
\frac{2}{r(u,v_2)} \int_{v_1}^{v_2} \frac{1}{2}\frac{ \theta^2}{\kappa}(u,\tilde v)d\tilde v \\ 
&\geq&
\nonumber
\frac{2}{r(u,v_2)} \int_{v_1}^{v_2} \partial_v\varpi(u,\tilde v)d\tilde v \\ 
&\geq& 
\label{etaDef}
\frac{2\left(\varpi(u,v_2)-\varpi(u,v_1)\right)}{r(u,v_2)}:=\eta(u)\;,
\end{eqnarray}
which we save for later use as 
\begin{equation}
\label{kBarQ}
 \frac{\underline\kappa(u,v_2)}{\underline\kappa(u,v_1)}\geq e^{\eta}\;.
\end{equation}

Now  notice that 
\begin{equation}
 \varpi\geq 0
\end{equation}
 since $\varpi|_{\Gamma}=(m-\Lambda r^3/6)|_{\Gamma}=0$ and, in view of~\eqref{dvpi}, $\partial_v\varpi\geq 0$. We then obtain the estimate (recall~\eqref{rBound})
\begin{equation}
\label{dr1-muEst}
\partial_r(1-\mu)(u,v)\geq -\frac{2\Lambda}{3} r(u,v)\geq -\frac{2\Lambda}{3} r(0,v_2)\;.
\end{equation}

Now in the region 
$${\cal R}_+ :=\{(u,v)\in {\cal R}\,:\, -\partial_r(1-\mu)(u,v)\geq 0\}=\{(u,v)\in {\cal R}\,:\, \frac{2\varpi}{r}\leq \frac{2\Lambda}{3} r^2\}$$
 we have
\begin{equation}
1 -\mu\geq 1-{\Lambda}r^2\geq 1-{\Lambda}r^2(0,v_2)
\end{equation}
 so that by imposing the condition $r(0,v_2)<\frac{1}{\sqrt{\Lambda}}$ (which we recall is one of the hypothesis of Theorem~\ref{mainThm}) we have  
\begin{equation}
\label{drQuotient}
\frac{-\partial_r(1-\mu)}{1-\mu} \leq \frac{\frac{2\Lambda}{3} r(0,v_2)}{ 1-{\Lambda}r^2(0,v_2)}\;,
\end{equation}
 first in ${\cal R}_+$ and then in the entire region $\cal R$, since the left hand side is negative in ${\cal R}\setminus {\cal R}_+$ and the right hand side is positive.
 
 We can now integrate~\eqref{dr} to obtain, for $0\leq u_1\leq u_2$,  
\begin{eqnarray}
 \nonumber
 \log\left(\frac{\lambda(u_2,v)}{\lambda(u_1,v)}\right)&=&\int_{u_1}^{u_2}{-\nu}\frac{-\partial_r(1-\mu)}{1-\mu} (u,v)du \\
 &\leq&
  \nonumber
 \frac{\frac{2\Lambda}{3} r(0,v_2)}{ 1-{\Lambda}r^2(0,v_2)} \int_{u_1}^{u_2}{-\nu}(u,v)du \\
  &=&
   \nonumber
 \frac{\frac{2\Lambda}{3} r(0,v_2)}{ 1-{\Lambda}r^2(0,v_2)} \left(r(u_1,v)-r(u_2,v)\right) \\
 &\leq&
 \label{alphaDef}
 \frac{\frac{2\Lambda}{3} r^2(0,v_2)}{ 1-{\Lambda}r^2(0,v_2)}=:\alpha_0\;,
\end{eqnarray}
which we save for future use in the form 
\begin{equation}
 \label{lambdaMonot}
 \lambda(u_2,v) \leq e^{\alpha_0}  \lambda(u_1,v) \;\;,\;\text{ for all } 0\leq u_1\leq u_2\;. 
\end{equation}
Integrating the last inequality, in $v$, readily gives
\begin{equation}
 \label{rDif}
r(u_2,v_2)-r(u_2,v_1) \leq e^{\alpha_0} \big( r(u_1,v_2)-r(u_1,v_1) \big)\;,
\end{equation}
for all $0\leq u_1\leq u_2$ and $0\leq v_1\leq v_2$.

Another consequence of~\eqref{lambdaMonot} and condition $(i)$ of Theorem~\ref{mainThm} is the following: let 
$(\bar u_1,v_1)\in \overline{\Gamma}\subset \overline{{\cal Q}^+}\setminus {\cal Q}^+$, then 
\begin{equation}
 \label{limRv1}
 \lim_{(u,v)\rightarrow (\bar u_1,v_1)\;,\; (u,v)\in J^-(\bar u_1,v_1)} r(u,v)=0\;.
\end{equation}
In fact, under such conditions we can consider a parameterization of $\Gamma$ of the form $(u,v_{\Gamma}(u))$, with $v_{\Gamma}(u)\rightarrow v_1$, as $u\rightarrow \bar u_1$. Then
 by integrating~\eqref{dr} from $\Gamma$, while recalling~\eqref{J-Gamma}, we get  
\begin{eqnarray*}
 r(u,v_1) = \int_{v_{\Gamma}(u)}^{v_1} \lambda(u,v)dv   
\leq e^{\alpha_0}\left[\sup_{v\leq v_1}\lambda(0,v)\right](v_1-v_{\Gamma}(u))\rightarrow 0\;, 
\end{eqnarray*}
 as $u\rightarrow u_1$, and~\eqref{limRv1} then follows from the monotonicity of $r$.
 
Relying once again on~\eqref{dr} and~\eqref{drQuotient}  gives
\begin{eqnarray}
 \nonumber
 \log\left(\frac{-\nu(u,v_2)}{-\nu(u,v_1)}\right)&=&\int_{v_1}^{v_2}{\lambda}\frac{\partial_r(1-\mu)}{1-\mu} (u,v)dv \\
 &\geq&
  \nonumber
 -\frac{\frac{2\Lambda}{3} r(0,v_2)}{ 1-{\Lambda}r^2(0,v_2)} \int_{v_1}^{v_2}{\lambda}(u,v)dv \\
  &=&
   \nonumber
 -\frac{\frac{2\Lambda}{3} r(0,v_2)}{ 1-{\Lambda}r^2(0,v_2)} \left(r(u,v_2)-r(u,v_1)\right) \\
 &\geq&
-\alpha_0\;,
\end{eqnarray}
 which implies  
\begin{equation}
 \label{nuMonot}
 -\nu(u,v_1) \leq - e^{\alpha_0}  \nu(u,v_2) \;\;,\;\text{ for all } 0\leq v_1\leq v_2\;. 
\end{equation}

We now turn to the scalar field's derivative and  using, in sequence,~\eqref{dz},~\eqref{kAlg}, H\"older's inequality,~\eqref{dvpi},~\eqref{kbarAlg},~\eqref{nuMonot} and the sign in~\eqref{dkbar}
\begin{eqnarray*}
 \left[\zeta(v_2,u)-\zeta(v_1,u)\right]^2
 &=&\left(\int_{v_1}^{v_2} \frac{-\nu\theta}{r}(u,v)dv\right)^2 \\
 &=&
 \left(\int_{v_1}^{v_2} \frac{-\nu\theta}{r} \frac{\sqrt{\lambda}}{\sqrt{\kappa(1-\mu)}}(u,v)dv\right)^2 \\ 
 &\leq &
 \int_{v_1}^{v_2} \frac{\theta^2}{\kappa}(u,v)dv\,  \int_{v_1}^{v_2} \frac{\nu^2\lambda}{r^2(1-\mu)}(u,v)dv\\
 &\leq&
 2\left[\varpi(u,v_2)-\varpi(u,v_1) \right] \int_{v_1}^{v_2} \nu \underline\kappa \frac{\lambda}{r^2}(u,v)dv \\
  &\leq&
  \eta  e^{\alpha_0} [r \nu \underline\kappa](u,v_2)  \int_{r(u,v_1)}^{r(u,v_2)} \frac{1}{r^2} dr \;,\\ 
\end{eqnarray*}
which, after introducing the quantity 
\begin{equation}
 \label{deltaDef}
 \delta=\delta(u):= \frac{r(u,v_2)-r(u,v_1)}{r(u,v_1)}\;, 
\end{equation}
reads
\begin{equation}
\label{zetaEstimate}
  \left[\zeta(v_2,u)-\zeta(v_1,u)\right]^2\leq  e^{\alpha_0}  \eta(u) \delta (u)  \nu(u,v_2)\underline\kappa(u,v_2) \;. 
  \end{equation}

To obtain our next estimate we introduce the notations
$$f_i=f_i(u):=f(u,v_i)\;\;, \; i=1,2\;,$$
and
$$\Delta f=f_2-f_1\;.$$
We also define 
\begin{equation}
 \label{uEta}
 u_{\eta}= \sup\{u\geq0\,:\,\exists v\geq 0 \;,(u,v)\in {\cal R} \text{ and }  \eta(u)>0 \}\;, 
 \end{equation}
 with the understanding that $u_{\eta}=+\infty$ if the defining set is empty. Note that, in view of~\eqref{mainAss}, $u_{\eta}>0$. 

The following estimates will be restricted to the set 
${\cal R}\cap \{u<u_{\eta}\}$;  later we will see that this is in fact the entire $\cal R$, i.e., that $u_{\eta}=+\infty$.
That being said, using~\eqref{kBarQ} and recalling that $\underline\kappa<0$, we have 
\begin{eqnarray*}
 \frac{\zeta^2}{\underline\kappa}(u,v_2) - \frac{\zeta^2}{\underline\kappa}(u,v_1) &=&  \Delta (\underline\kappa^{-1}\zeta^2)\\
 &=&    \underline\kappa_2^{-1}\left[\zeta_2^2-\underline\kappa_2\underline\kappa_1^{-1}\zeta_1^2\right]   \\
 &=& \underline\kappa_2^{-1}\left[(\Delta \zeta)^2+2\zeta_1\zeta_2-\zeta_1^2-\underline\kappa_2\underline\kappa_1^{-1}\zeta_1^2\right] \\
 &=& \underline\kappa_2^{-1}\left[(\Delta \zeta)^2+2(\Delta\zeta) \zeta_1 +\zeta_1^2-\underline\kappa_2\underline\kappa_1^{-1}\zeta_1^2\right] \\
&\geq& \underline\kappa_2^{-1}\left[(\Delta \zeta)^2+2(\Delta\zeta) \zeta_1 -(e^{\eta}-1)\zeta_1^2\right] \\
&=& \underline\kappa_2^{-1}\left[(\Delta \zeta)^2+(e^{\eta}-1)\left(\frac{2(\Delta\zeta) \zeta_1}{e^{\eta}-1} -\zeta_1^2\right)\right] \\
&\geq& \underline\kappa_2^{-1}\left[(\Delta \zeta)^2+(e^{\eta}-1)\left(\frac{(\Delta\zeta)^2}{(e^{\eta}-1)^2}+ \zeta_1^2 -\zeta_1^2\right)\right]\;.
\end{eqnarray*}
In conclusion, in ${\cal R}\cap \{u<u_{\eta}\}$, we have
\begin{equation}
\label{zetaKappa}
  \frac{\zeta^2}{\underline\kappa}(u,v_2) - \frac{\zeta^2}{\underline\kappa}(u,v_1)
   \geq \frac{1}{\underline\kappa(u,v_2)}\left(1+\frac{1}{e^{\eta}-1}\right) \left[\zeta(u,v_2) - \zeta(u,v_1)\right]^2\;.
\end{equation}

Noting that we can rewrite~\eqref{etaDef} as
$$\eta=\frac{2\Delta\varpi}{r_2}$$
we have, using~\eqref{dupi},~\eqref{zetaKappa},~\eqref{zetaEstimate},
\begin{eqnarray*}
 \frac{d\eta}{du}&=&\frac{1}{r_2}\Delta\left(\frac{\zeta^2}{\underline\kappa}\right)-\frac{2\Delta\varpi}{r_2^2}\nu_2\\
 &\geq& \frac{1}{r_2\underline\kappa_2} \left(1+\frac{1}{e^{\eta}-1}\right) (\Delta\zeta)^2+\frac{-\nu_2}{r_2}\eta\\
  &\geq& \frac{1}{r_2\underline\kappa_2} \left(1+\frac{1}{e^{\eta}-1}\right) e^{\alpha_0} \eta \delta   \underline\kappa_2 \nu_2 +\frac{-\nu_2}{r_2}\eta
\end{eqnarray*}
which we write as 
\begin{equation}
\label{dedu}
 \frac{d\eta}{du}\geq \frac{-\nu(u,v_2)}{r(u,v_2)}\eta\left[ 1- e^{\alpha_0} \delta \left(1+\frac{1}{e^{\eta}-1}\right)\right]\;.
\end{equation}

Exactly as in~\cite{Christodoulou91} we introduce a new variable
\begin{equation}
 x= \frac{r(u,v_2)}{r(0,v_2)}\;,
\end{equation}
which, after noting that $\delta_0=\delta(0)$ and recalling~\eqref{rDif}, we use to obtain 
\begin{eqnarray*}
 \delta&=&\frac{r(u,v_2)-r(u,v_1)}{r(u,v_1)} =\frac{r(u,v_2)-r(u,v_1)}{r(u,v_2)-\left[(r(u,v_2)-r(u,v_1)\right]}  \\
 &\leq & \frac{e^{\alpha_0}\left[(r(0,v_2)-r(0,v_1)\right]}{r(u,v_2)-e^{\alpha_0}\left[(r(0,v_2)-r(0,v_1)\right]} \\
  &= & \frac{e^{\alpha_0} \delta_0}{\frac{r(u,v_2)}{r(0,v_1)}-e^{\alpha_0} \delta_0} 
  =\frac{e^{\alpha_0} \delta_0}{x\frac{r(0,v_2)}{r(0,v_1)}-e^{\alpha_0} \delta_0} \;,
\end{eqnarray*}
from which we get 
\begin{equation}
\label{deltaEst}
 \delta\leq \frac{e^{\alpha_0} \delta_0}{x(\delta_0+1)-e^{\alpha_0} \delta_0}\;.
\end{equation}

In the new variable we have
$$\frac{d\eta}{dx}=\frac{d\eta}{du}\frac{r(0,v_2)}{\nu(u,v_2)}$$
and then, using~\eqref{dedu}, the elementary inequality $e^{\eta}-1\geq \eta$ and~\eqref{deltaEst}, we arrive at
\begin{equation}
\label{difIneq}
\frac{d\eta}{dx} +\frac{(1-f)}{x}\eta\leq\frac{f}{x}\;,
\end{equation}
with
\begin{equation}
f:=\frac{e^{2\alpha_0} \delta_0}{x(\delta_0+1)-e^{\alpha_0} \delta_0}\;. 
\end{equation}

Let
\begin{equation}
 \label{x0Def}
 x_0:=\frac{e^{\alpha_0}\delta_0}{1+\delta_0}\;,
\end{equation}
which will be made smaller than unity by future restrictions on $\delta_0$.
Then we can easily integrate~\eqref{difIneq} and obtain, for $x\in(x_0,1]$,
\begin{equation}
\label{etaEst}
 \eta(x)\geq e^{G(x)}(\eta_0-F(x))\;,
\end{equation}
where
\begin{equation}
 G(x)=\int_x^1 \frac{1-f(y)}{y}dy\;,
\end{equation}
and 
\begin{equation}
 F(x)=\int_x^1 \frac{f(y)}{y}e^{-G(y)}dy\;.
\end{equation}

Now let 
$$H(x)=F(x)+e^{-G(x)}$$
and note that 
$$\frac{dH}{dx} = \frac{e^{-G}}{x}(1-2f)\;.$$
We can then conclude that 
\begin{equation}
\label{x1}
 x_1=x_1(\delta_0)=\frac{(2e^{\alpha_0}+1)e^{\alpha_0}\delta_0}{1+\delta_0}
\end{equation}
is the only critical point of $H$ and that moreover 
$$H'(x)<0\;\;,\; \text{ for } x<x_1\;,$$
 and 
$$  H'(x)>0\;\;,\; \text{ for } x>x_1\;.$$

It is now essential to note that $\alpha_0$ decreases with $\delta_0$ (just recall~\eqref{alphaDef} and that $r(0,v_2)=(1+\delta_0)r(0,v_1)$); then it becomes clear that  there exists 
$\delta_1= \delta_1(\Lambda)$ such that 
\begin{equation}
x_1(\delta_0)<1\;\;, \text{ for all } \delta_0\leq  \delta_1\;.
\end{equation}
In conclusion, for $\delta_0<\delta_1$, $x_1$ is the global minimum of $H$ in $(x_0,1]$.  We then define 
\begin{equation}
 E(\delta_0):=H(x_1(\delta_0))\;.
\end{equation}

A direct computation gives rise to the elementary formula
\begin{equation}
\label{GExplicit}
G(x)=-\log(x)-e^{\alpha_0}\log\left[\frac{(\delta_0+1-e^{\alpha_0}\delta_0)x}{(\delta_0+1)x-e^{\alpha_0}\delta_0}\right]\;,
\end{equation}
from which we immediately see that, for $\delta_0<\delta_1$ (after decreassing $\delta_1$ if necessary), we have, for appropriate choices of $0<c_1<C_1$,
\begin{equation}
\label{e^G1Estimate}
c_1 \delta_0 \leq e^{-G(x_1(\delta_0))}\leq C_1 \delta_0\;.
\end{equation}
However we are unaware of the existence of an elementary formula for $F$; note, in contrast, that in the $\Lambda=0$ case we obtain $e^{\alpha_0}=1$ and then such a formula can be easily obtained~\cite{Christodoulou91}.

Nonetheless we have the following estimate: 
\begin{lemma}
 \label{lemFEstimate}
 There are positive constants $c_1,C_1$ and $\delta_1$, such that for   $\delta_0<\delta_1$
\begin{equation}
\label{FEstimate}
c_1\delta_0\log(1/\delta_0)\leq  F(x_1(\delta_0)) \leq C_1\delta_0\log(1/\delta_0)\;.
\end{equation}
\end{lemma}
\begin{proof}
 Using~\eqref{GExplicit} we get 
 $$F(x)=\frac{e^{2\alpha_0}\delta_0\left(\delta_0+1-e^{\alpha_0}\delta_0\right)^{e^{\alpha_0}}}{\left(\delta_0+1\right)^{e^{\alpha_0}+1}}
\int_x^1\frac{y^{e^{\alpha_0}}}{\left(y-\frac{e^{\alpha_0}\delta_0}{\delta_0+1}\right)^{e^{\alpha_0}+1}}dy$$

 If we set $n=\min\{m\in\mathbb{Z}\;:\; m\geq e^{\alpha_0} \}\geq 2$ and write $a=\frac{e^{\alpha_0}\delta_0}{\delta_0+1}$ we see that, since $\frac{y}{y-a}> 1$ then 
 %
 %
\begin{eqnarray*}
 \int_x^1\frac{y^{e^{\alpha_0}}}{\left(y-\frac{e^{\alpha_0}\delta_0}{\delta_0+1}\right)^{e^{\alpha_0}+1}}dy 
 &\leq& \int_x^1\frac{y^n}{\left(y-a\right)^{n+1}}dy \\
 &=& \int_x^1\frac{\sum_{k=0}^n {n\choose k} (y-a)^k a^{n-k}}{\left(y-a\right)^{n+1}}\,dy \\
 &=& \int_x^1\frac{1}{(y-a)}+\sum_{k=0}^{n-1} {n\choose k} (y-a)^{k-n-1} a^{n-k}\,dy \\
 &=& \left[ \log(y-a)-\sum_{k=0}^{n-1} \frac{{n\choose k}a^{n-k}}{n-k} \frac{ 1}{\left(y-a\right)^{n-k}}\right]_{x}^1
\end{eqnarray*}
 and the desired upper bound follows by evaluating the previous expression at $x=x_1$. The lower bound is similar although easier (one just needs to use $\left(\frac{y}{y-a}\right)^{e^{\alpha_0}}> \frac{y}{y-a}$). 
\end{proof}

We then conclude that we also have
\begin{equation}
\label{EEstimate}
c_1\delta_0\log(1/\delta_0)\leq  E(\delta_0) \leq C_1\delta_0\log(1/\delta_0)\;.
\end{equation}

In particular, from~\eqref{mainAss} we see that $\eta_0>E(\delta_0)>F(x)$, for all $x$, so that~\eqref{etaEst} implies that $u_{\eta}=+\infty$ and the derived estimates hold in the entire region $\cal R$.

\vspace{0,5cm}

We are now ready to prove that under assumption~\eqref{mainAss} a trapped surface has to form along $v=v_2$ . We will argue by contradiction by assuming that the entire line $v=v_2$ is contained in ${\cal R}$, i.e, that 
\begin{equation}
\label{contraAss}
 [0,\bar{u}_2)\times\{v_2\}\subset {\cal R}  \text{ and } (\bar{u}_2,v_2)\in \overline{{\cal Q}^+}\setminus {\cal Q}^+\;.
\end{equation}
Note that using~\eqref{Rpast} we then have  $J^-([0,\bar{u}_2)\times\{v_2\})\subset \cal R$ and, in particular, 
$[0,\bar{u}_2)\times [v_1,v_2]\subset\cal R$. 

Then, relying on~\eqref{rDif}, we see that  
\begin{eqnarray*}
 r(u,v_1) &=& r(u,v_2) -  (r(u,v_2)-r(u,v_1))  \\
 &\geq& x(u)r(0,v_2)-e^{\alpha_0}(r(0,v_2)-r(0,v_1)) \\
& =& x(u)(1+\delta_0)r(0,v_1) -e^{\alpha_0}\delta_0 r(0,v_1)
\end{eqnarray*}
which reads 
\begin{equation}
\label{rEstX} 
 r(u,v_1)\geq \left[x(u)(1+\delta_0)-e^{\alpha_0}\delta_0 \right] r(0,v_1)\;.
\end{equation}
We will now show that 
\begin{equation}
\label{limX}
 \lim_{u\rightarrow \bar{u}_2} x(u)\leq x_0\;,
\end{equation}
for $x_0$ given by~\eqref{x0Def}. In fact, by monotonicity the limit clearly exists and if we assume that the desired bound is not satisfied, then 
$x(u)>x_0+\epsilon$, for some $\epsilon>0$ and for all $u<\bar u_2$. After noticing that $(u,v_2)\in{\cal Q}^+\Rightarrow (u,v_1)\in{\cal Q}^+$,~\eqref{rEstX} gives 
$$\lim_{u\rightarrow\bar u_2} r(u,v_1)>0$$
 which in view of assumption $(i)$ of Theorem~\ref{mainThm} and~\eqref{limRv1} implies that $(\bar u_2,v_1)\in {\cal Q}^+\setminus\Gamma$. Now let $\bar v=\sup\{v\in [v_1,v_2)\;:\; (\bar u_2, v)\in{\cal Q}^+\}>v_1$  and note that 
 $$r(u,v)>r(\bar u_2,v_1)>0\;\;,\;\text{ for all } (u,v)\in[0,\bar u_2]\times [v_1,\bar v]\setminus \{(\bar u_2,\bar v)\}\;.$$
 Applying the extension principle of Proposition~\ref{propExt} we conclude that $(\bar u_2,\bar v)\in  {\cal Q}^+$ and therefore $\bar v = v_2$ which contradicts the fact that $(\bar u_2,v_2)\notin  {\cal Q}^+$.  In conclusion,~\eqref{limX} holds.

%
We are now allowed to set  $u_0=\lim_{x\rightarrow x_0}u(x)$ and then $x>x_0\Leftrightarrow u<u_0$. In particular, under assumption~\eqref{contraAss}, we see that $(u,v_2)\in \cal R$, for all $u<u_0$. As a  consequence, since $u_1:=u(x_1)<u_0$, we conclude that $(u_1,v_2)\in\cal R$. 

As seen before~\eqref{less1}, in $\cal R$, we have $1-\mu>0$ which implies that 
$$\frac{2\varpi}{r}<1-\frac{\Lambda}{3}r^2\leq 1\,,$$
from which we readily conclude that, in $\cal R$,  
\begin{equation}
\label{etaR}
\eta<1\;. 
\end{equation}
But then~\eqref{contraAss} implies that $\eta(x_1)<1$ which, in view of~\eqref{etaEst}, leads to  
\begin{equation}
\label{ineqContra}
 \eta_0\leq e^{-G(x_1)}\eta(x_1) +F(x_1) < e^{-G(x_1)}+F(x_1)= E(\delta_0)\;,
\end{equation}
in contradiction with~\eqref{mainAss}, in view of~\eqref{EEstimate}.  It then follows that~\eqref{contraAss} cannot be true and therefore there must exist $u^*<\bar{u}$ such that $\lambda(u^*,v_2)=0$.

In fact as a consequence of~\eqref{lambdaNeg}
$u^*$ is the unique solution to $\lambda(u,v_2)=0$. Moreover, recalling the argument leading to~\eqref{etaR} and~\eqref{ineqContra}, we must have  $x(u^*)\geq x_1$ and then 
\begin{eqnarray}
\nonumber
r(u^*,v_2)&=&x(u^*)r(0,v_2) \\ 
\nonumber
&\geq&  x_1r(0,v_2) \\ 
\nonumber
&=& \frac{(2e^{\alpha_0}+1)e^{\alpha_0}\delta_0}{1+\delta_0}r(0,v_2)\\
\label{r*Est}
&=&{(2e^{\alpha_0}+1)e^{\alpha_0}\delta_0}r(0,v_1)\;.
\end{eqnarray}

Moreover, again by Proposition~\ref{propGlobal} we also have $[0,u^*)\times\{v_2\}\subset\cal R$ and we are allowed to apply~\eqref{rEstX} to obtain, for all $u<u^*\leq u_1$,
\begin{eqnarray}
\nonumber
 r(u,v_1)&\geq&  \left[x(u)(1+\delta_0)-e^{\alpha_0}\delta_0 \right] r(0,v_1) \\
 \nonumber
 &\geq&  \left[x_1(1+\delta_0)-e^{\alpha_0}\delta_0 \right] r(0,v_1) \\
 &\geq& 2e^{2\alpha_0} \delta_0r(0,v_1)\;.
\end{eqnarray}
By taking the limit when $u\rightarrow u^*$ we see that
\begin{equation}
 r(u^*,v_1)\geq 2e^{2\alpha_0} \delta_0r(0,v_1)\;,
\end{equation}
which, in view of~\eqref{limRv1}, allows us to conclude that the marginally trapped surface forms, along $v=v_2$, before (as measured by $u$) $v=v_1$ reaches $\bar \Gamma$. 
\end{proof}

\section*{Acknowledgements}

This work was partially supported by FCT/Portugal through UID/MAT/04459/2013 and by FCT/Portugal and CERN through CERN/FIS-PAR/0023/2019. 

%


\begin{thebibliography}{}

\bibitem{An12} X.~An,~\emph{Formation of Trapped Surfaces from Past Null Infinity}, arXiv:1207.5271 [gr-qc].

\bibitem{An17} X.~An and J.~Luk,~\emph{Trapped surfaces in vacuum arising from mild incoming radiation}, Adv. Theo. Math. Phys., {\bf 21}  (2017) 1, 1-120.

\bibitem{An20} X.~An and R.~Zhang,~\emph{Polynomial Blow-Up Upper Bounds for the Einstein-Scalar Field System Under Spherical Symmetry}, Commun. Math. Phys. (2020). https://doi.org/10.1007/s00220-019-03677-0

\bibitem{An202} X.~An and D.~Gajic,~\emph{Curvature blow-up rates in spherically symmetric gravitational collapse to a Schwarzschild black hole}, arXiv:2004.11831. 

\bibitem{Andreasson17} 
H.~Andr\'easson and G.~Rein,~\emph{Formation of trapped surfaces for the spherically symmetric Einstein-Vlasov system}, J. Hyperbol. Diff. Equat. {\bf 7} (2010) 707-731.

\bibitem{Burtscher17} 
A.~Burtscher and P.~G.~LeFloch,~\emph{The formation of trapped surfaces in spherically-symmetric Einstein-Euler spacetimes with bounded variation}, Journal de Mathématiques Pures et Appliquées {\bf 102} (2014), 6, 1164–1217. 

\bibitem{Christodoulou91}
D.\ Christodoulou, \emph{The Formation of black holes and singularities in spherically symmetric gravitational collapse}, Comm. Pure App. Math. {\bf 64}  (1991) 339-373

\bibitem{Christodoulou95}
D.\ Christodoulou, \emph{Self-gravitating relativistic fluids: a two-phase model}, Arch. Rational Mech. Anal. {\bf 130} (1995) 343–400.

\bibitem{Christodoulou99}
D.\ Christodoulou, \emph{The instability of naked singularities in the gravitational collapse of a scalar field}, Ann. Math. {\bf 149} (1999) 183–217.

\bibitem{Christodoulou09}
D.\ Christodoulou, \emph{The Formation of Black Holes in General Relativity}, EMS Monographs in Mathematics (2009).

\bibitem{Chrusciel09} P. T. Chru\'sciel, G. J. Galloway, and D. Solis, \emph{Topological censorship for Kaluza-Klein space-times}, Ann. Henri Poincar\'e {\bf 10} (2009) 5, 893-912.

\bibitem{Chrusciel18} P. T. Chru\'sciel, G. J. Galloway, and E. Ling, \emph{Weakly trapped surfaces in asymptotically de Sitter spacetimes}, Class. Quant. Grav. {\bf 35} (2018) 13, 135001.
 
\bibitem{CGNS1}
{J.~L~.~Costa, P.~Gir\~ao, J.~Nat\'ario and J.~Silva},  \emph{On the global uniqueness for the Einstein-Maxwell-scalar field system with a cosmological constant.
Part 1.  Well posedness and breakdown criterion},  Class.\ Quantum Grav.\ {\bf 32} (2015) 015017.

\bibitem{CGNS2} {J.~L.~Costa, P.~Gir\~ao, J.~Nat\'ario and J.~Silva},  \emph{On the global uniqueness for the Einstein-Maxwell-scalar field system with a cosmological constant. Part 2. Structure of the solutions and stability of the Cauchy horizon}, Commun.\ Math.\ Phys.\ {\bf 339} (2015) 903-947.

\bibitem{CGNS3} {J.~L.~Costa, P.~Gir\~ao, J.~Nat\'ario and J.~Silva},  \emph{On the global uniqueness for the Einstein-Maxwell-scalar field system with a cosmological constant. Part 3: Mass inflation and extendibility of the solutions}, Ann. PDE {\bf 3} (2017) 8.

\bibitem{CGNS4}
{J.~L.~Costa, P.~Gir\~ao, J.~Nat\'ario and J.~Silva},  \emph{On the occurrence of mass inflation for the Einstein-Maxwell-scalar field system with a cosmological constant and an exponential Price law},   Commun. Math. Phys. {\bf 361} (2018) 289-341.

\bibitem{CostaProblem}  J.~L.~Costa, A.~Alho and J.~Nat\'ario , \emph{The problem of a self-gravitating scalar field with positive cosmological constant'}, {\em Ann. Henri Poincar\'e}, {\bf 14} (2012) 1077--1107.

\bibitem{CostaMena}  J.~L.~Costa and F.~Mena, \emph{Global solutions to the spherically symmetric Einstein-scalar field system with a positive cosmological constant in Bondi coordinates}, arXiv:2004.07396 [gr-qc].

\bibitem{Costa-Natario-Oliveira} J.~L.~Costa, J.~Nat\'ario and P.~Oliveira, \emph{Cosmic no-hair in spherically symmetric black hole spacetimes},  {\em Ann. Henri Poincar\'e}, {\bf 20} (2019) 3059--3090.

\bibitem{Dafermos04}
{M.~Dafermos}, \emph{Stability and instability of the Reissner-Nordstrom Cauchy horizon and the problem of uniqueness in general relativity}, Contemp. Math. {\bf 350} (2004), 99–113

\bibitem{Dafermos05}
{M.~Dafermos},  \emph{Spherically symmetric space-times with a trapped surface},   Class.Quant.Grav. {\bf 22} (2005) 2221-2232.

\bibitem{DafRen}
M.\ Dafermos and A.\ Rendall \emph{Strong cosmic censorship for surface-symmetric cosmological spacetimes with collisionless matter}, Comm.\ Pure Appl.\ Math.\ {\bf 69} (2016) 815–908.

\bibitem{hintzVasy} P.~Hintz  and A.~Vasy, ``The global non-linear stability of the Kerr-de Sitter family of black holes", {\em Acta Mathematica},  {\bf 220} (2018) 1-206.

\bibitem{Hintz} P.~Hintz, ``Non-linear stability of the Kerr-Newman-de Sitter family of charged black holes", 
{\em Annals of PDE}, {\bf 4} (2018) (1):11.

\bibitem{Klainerman12}
{S.~Klainerman and I.~Rodnianski},  \emph{On the formation of trapped surfaces},   Acta Math., {\bf 208} (2012) 211–333.

\bibitem{Klainerman14}
{S.~Klainerman, J.~Luk and I.~Rodnianski},  \emph{A fully anisotropic mechanism for formation of trapped surfaces in vacuum},   Inventiones mathematicae, {\bf 198} (2014) 1–26.

\bibitem{Kommemi}
{J.~Kommemi},  \emph{The Global structure of spherically symmetric charged scalar field spacetimes
},   Commun.Math.Phys. {\bf 323} (2013) 35-106

\bibitem{Luk19} J.~Luk, S.-J.~Oh and S.~Yang, \emph{Dynamical black holes with prescribed masses in spherical symmetry}, Proceedings of 7th ICCM II (2019) 367-387.

\bibitem{schlue2}
V.~Schlue, \emph{Decay of the Weyl curvature in expanding black hole cosmologies}, arXiv:1610.04172.


\end{thebibliography}
\end{document}